\documentclass[conference]{IEEEtran}
\IEEEoverridecommandlockouts

\usepackage{cite}
\usepackage{amsmath,amssymb,amsfonts}
\usepackage{algorithmic}
\usepackage{graphicx,enumitem,kantlipsum}
\usepackage{textcomp}
\usepackage{xcolor}
\usepackage{mathtools}
\mathtoolsset{showonlyrefs}
\def\BibTeX{{\rm B\kern-.05em{\sc i\kern-.025em b}\kern-.08em
    T\kern-.1667em\lower.7ex\hbox{E}\kern-.125emX}}

\usepackage{bbm}

\usepackage{tikz}
\usetikzlibrary{arrows.meta}

\newtheorem{theorem}{Theorem}
\newtheorem{proposition}{Proposition}
\newtheorem{lemma}{Lemma}

\newtheorem{remark}{Remark}
\newtheorem{corollary}{Corollary}

\def\fndot{\, \cdot \,}
\newenvironment{proof}[1][Proof]
  {\begin{IEEEproof}[#1]}
  {\end{IEEEproof}}

\begin{document}

\title{On the Tightness of Omura's Strong Converse Exponent and Critical Rates Above Capacity
\thanks{The work of H.~Joudeh was supported by the European Research Council (Grant No.~101116550). The work of V.~Y.~F.~Tan was supported by a Singapore Ministry of Education AcRF Tier~2 grant (A-8004062-00-00) and a Singapore National Research Foundation Investigatorship (A-8004657-00-00).}
}

\author{
\IEEEauthorblockN{Hamdi Joudeh}
\IEEEauthorblockA{Department of Electrical Engineering \\
Eindhoven University of Technology \\
\texttt{h.joudeh@tue.nl}} \and  \IEEEauthorblockN{Vincent Y. F. Tan}
\IEEEauthorblockA{Department of Mathematics and Department of ECE \\
National University of Singapore\\
\texttt{vtan@nus.edu.sg}}
}

\maketitle

\begin{abstract}
Omura's lower bound on the probability of correct decoding at rates above capacity has an exponent of sphere packing form. Dueck and K\"orner refined Omura's argument through a codebook extension step and obtained the exact strong converse exponent.
It has been claimed in the literature that the two exponents coincide at all rates, but the argument contains an error. 
To settle this question, we first show that above a threshold rate, the strong converse exponent follows a straight line of unit slope determined by the R\'enyi capacity of order infinity. 
This threshold rate plays the role of a critical rate above capacity, while the order-infinity capacity plays that of a cutoff rate. 
We then show that the two exponents coincide up to a second threshold rate, which lies above the first, and differ strictly beyond it.  Examples include modulo-additive channels, where there is no gap; and the binary erasure channel, where the thresholds and the gap are computed in closed form.
\end{abstract}
\section{Introduction}
\label{sec:introduction}
Consider a discrete memoryless channel (DMC) with input alphabet $\mathcal{X}$ and output alphabet $\mathcal{Y}$, both finite, and transition probability matrix $W$. Let $I(P,W)$ denote the mutual information between an input $X$ and the induced output $Y$, where $(X,Y) \sim P \times W$ for some input distribution $P$. The capacity of this channel is given by $C(W) := \max_{P} I(P,W)$.

Let $p_{\mathrm{e}}^{\star}(n,R)$ be the average probability of error of an optimal code of blocklength $n$ and rate at least $R$.
The strong converse theorem for the DMC, first proved by Wolfowitz \cite{Wolfowitz1957}, asserts that $p_{\mathrm{e}}^{\star}(n,R) \to 1$ as $n \to \infty$ for all rates above capacity.
Later, Arimoto \cite{Arimoto1973} showed that this convergence is at least exponentially fast by establishing the lower bound
\begin{equation}
    p_{\mathrm{e}}^{\star}(n,R) \geq 1 - \exp \left( - n K_{\mathrm{A}}(R) \right),
\end{equation}
where  Arimoto's strong converse exponent $K_{\mathrm{A}}(R)$ is given by
\begin{equation*}
    K_{\mathrm{A}}(R) :=  \max_{\rho \in [-1,0]} \left\{  - \rho R + \min_{P} E_0(\rho,P) \right\},
\end{equation*}
and $E_0(\rho,P)$ is Gallager's function
\begin{equation}
   E_0(\rho,P)  := - \log \sum_{y \in \mathcal{Y}} \bigg( \sum_{x \in \mathcal{X}} P(x) W(y|x)^{\frac{1}{1+\rho}} \bigg)^{1+\rho},
\end{equation}
which was first introduced in the characterization of error exponents at rates below capacity \cite{Gallager1968}. Arimoto further noted that $K_{\mathrm{A}}(R) > 0$ if and only if $R > C(W)$, which recovers Wolfowitz's strong converse. His proof exploits the invariance of the error probability under permutations of the codebook, together with the analytical properties (mainly convexity) of the $E_0$ function established by Gallager \cite{Gallager1968}.
An alternative proof of Arimoto's strong converse theorem appears in \cite{Polyanskiy2010}.

Arimoto's bound naturally raises the question of the exact exponential rate at which $p_{\mathrm{e}}^{\star}(n,R)$ approaches one. This rate is the strong converse exponent for channel coding, also known as the reliability function above capacity, and is defined as
\begin{equation}
    K_{\mathrm{sc}}(R) := \limsup_{n \to \infty} - \frac{1}{n} \log \big( 1 - p_{\mathrm{e}}^{\star}(n,R) \big). \label{eqn:Ksc}
\end{equation}
In this terminology, Arimoto's result gives a lower bound $K_{\mathrm{sc}}(R) \geq K_{\mathrm{A}}(R)$. Omura \cite{Omura1975} derived an upper bound (i.e., achievable exponent) given by $K_{\mathrm{sc}}(R) \leq K_{\mathrm{O}}(R)$, where\footnote{Omura's argument does not fix the composition of the code. The form with a minimum over $P$ is obtained by fixing the composition, as in \cite[Lemma~2]{Dueck1979}.}
\begin{equation}
    K_{\mathrm{O}}(R) := \min_{P} \min_{V: I(P,V) \geq R} D(V \| W | P).
    \label{eqn:K_O}
\end{equation}
Omura obtained this bound using a change-of-measure argument applied to DMC $W$, introduced by Haroutunian~\cite{Haroutunian1968}; see also~\cite{Csiszar2011}.
It is readily seen that $K_{\mathrm{O}}(R)>0$ if and only if $R>C(W)$. The variational formulas for $K_{\mathrm{A}}(R)$ and $K_{\mathrm{O}}(R)$ do not immediately reveal when equality holds. Omura asserted in~\cite{Omura1975} that the two exponents coincide for rates sufficiently close to capacity from above.

The strong converse exponent was determined by Dueck and K\"orner \cite{Dueck1979}, who showed that the $\limsup$ in \eqref{eqn:Ksc} is a limit and that $K_{\mathrm{sc}}(R) = K_{\mathrm{DK}}(R)$, where
\begin{equation}
    K_{\mathrm{DK}}(R) := \min_{P,V} \big\{ D(V \| W | P) + |R - I(P,V)|^{+} \big\}.\label{eqn:K_DK}
\end{equation}
Note that $K_{\mathrm{DK}}(R) \leq K_{\mathrm{O}}(R)$, since the rate constraint $I(P,V) \geq R$ in Omura's exponent is replaced by the penalty term $|R - I(P,V)|^{+} := \max \{0, R - I(P,V)\}$. The upper bound (achievability part) refines Omura's argument through a codebook extension trick, while the lower bound (converse part) is based on the method of types.
Dueck and K\"orner further asserted, without proof, that $K_{\mathrm{DK}}(R) = K_{\mathrm{A}}(R)$.

In \cite{oohama2012}, Oohama set out to show that $K_{\mathrm{DK}}(R) = K_{\mathrm{A}}(R)$. His argument passes through the intermediate claim that
\begin{equation}
\label{eq:K_DK_K_O_equality}
K_{\mathrm{DK}}(R) = K_{\mathrm{O}}(R) \quad \text{for all} \ R \leq \log |\mathcal{X}|,
\end{equation}
which implies that $K_{\mathrm{O}}(R)$ is tight in this range; see \cite[Property 3(b)]{oohama2012}. The proof of this claim, carried out for a fixed $P$, rests on the assertion that $D(V \| W | P) - I(P,V)$ is  linear in $V$. This function, however, is in fact convex but not linear.\footnote{This gap was noticed in \cite{Tan2025} and confirmed by Prof.\ Oohama to the authors in a private communication.}
Nevertheless, the equality $K_{\mathrm{DK}}(R) = K_{\mathrm{A}}(R)$ does hold, as originally claimed by Dueck and K\"orner, and later confirmed by other means \cite{MosonyiOgawa2017}; see also Section~\ref{sec:dual}.

Determining the full range of equality between $K_{\mathrm{DK}}(R)$ and $K_{\mathrm{O}}(R)$ is a question of interest in its own right,    and one that is not settled by  $K_{\mathrm{DK}}(R) = K_{\mathrm{A}}(R)$.
It concerns the proofs of Dueck and K\"orner \cite[Lemma~2]{Dueck1979} and of Mosonyi and Ogawa~\cite[Lemma~V.11]{MosonyiOgawa2017}, since if \eqref{eq:K_DK_K_O_equality} held, Omura's bound would already be tight and the codebook extension step in their achievability arguments would not be needed.
It is also the strong converse counterpart of  whether the sphere packing and random coding exponents coincide at rates below capacity. The latter has a simple answer, since both exponents are convex in $R$ and agree precisely where the sphere packing exponent has slope no steeper than $-1$, that is, above the critical rate. No such structure is available above capacity: unlike its sphere packing counterpart, $K_{\mathrm{O}}(R)$ is not known to be convex in $R$, as the constraint set $\{V : I(P,V) \geq R\}$ is not convex.

For constant composition codes, Cheng and Nakibo{\u{g}}lu~\cite{ChengNakiboglu2020} proved a refined strong converse bound with polynomial prefactors. Their tilted-channel characterization~\cite[Eqs.~(25)--(28)]{ChengNakiboglu2020} implies that, for a fixed type $P$, the Omura and Dueck--K\"orner exponents coincide up to the onset of the linear branch. Their results do not characterize how far this equality persists along the branch. Here, we establish the range of equality for the exponents optimized over the input distribution, including how far equality persists along this branch.
\subsubsection*{Contributions}
Our main contributions are as follows. 
\begin{list}{}{\setlength{\leftmargin}{1.6em}\setlength{\labelwidth}{1.2em}\setlength{\labelsep}{0.4em}\setlength{\itemsep}{2pt}\setlength{\parsep}{0pt}\setlength{\topsep}{2pt}}
    \item[(i)] We provide a necessary and sufficient condition for $K_{\mathrm{DK}}(R) = K_{\mathrm{O}}(R)$ in terms of the growth of $K_{\mathrm{O}}$.
    \item[(ii)] We show that the two exponents coincide on an interval $[0, R^{\star}]$ and differ strictly on $(R^{\star}, \log|\mathcal{X}|]$, where $R^{\star}$ is given explicitly in terms of the maximum-likelihood structure of $W$ (Theorem~\ref{thm:main}). The proof identifies the role of the endpoint line $R - C_{\infty}(W)$ of Arimoto's family, which $K_{\mathrm{DK}}$ follows from a second threshold $R_{\star} \leq R^{\star}$ onwards. We observe that $R_{\star}$ and $C_{\infty}(W)$ play the roles of a critical rate and a cutoff rate above capacity.
    \item[(iii)] We characterize DMCs for which there is no gap at any rate, and compute the rate thresholds  in closed form for modulo-additive channels and the binary erasure channel.
\end{list}
\section{A Criterion for $K_{\mathrm{DK}} = K_{\mathrm{O}}$}
We start our exposition with a necessary and sufficient condition for equality between $K_{\mathrm{DK}}(R)$ and $K_{\mathrm{O}}(R)$. To this end, the following inf-convolution identity is instrumental.
\begin{lemma}\label{lem:inf_conv}
For every $R \geq 0$,
\begin{equation}
    K_{\mathrm{DK}}(R) = \min_{0 \leq a \leq R}
    \big\{ K_{\mathrm{O}}(a) + (R - a) \big\}.
    \label{eq:inf_conv}
\end{equation}
\end{lemma}
\begin{proof}
Let $(P,V)$ attain the minimum defining $K_{\mathrm{DK}}(R)$ and set $ a_{P,V} := \min \{ I(P,V), \, R \} $.
It can be readily verified that $|R - I(P,V)|^{+} = R - a_{P,V}$.
Since $I(P,V) \geq a_{P,V}$, the pair $(P,V)$ is feasible in the minimization defining $K_{\mathrm{O}}(a_{P,V})$, and $D(V \| W | P) \geq K_{\mathrm{O}}(a_{P,V})$.
It follows that
\begin{align*}
 K_{\mathrm{DK}}(R)   & = D(V \| W | P) + |R - I(P,V)|^{+} \\
    & \geq K_{\mathrm{O}}(a_{P,V}) + (R - a_{P,V}) \\
    & \geq \inf_{0 \leq a \leq R} \big\{ K_{\mathrm{O}}(a) + (R - a) \big\}.
\end{align*}

For the other direction, fix $a \in [0,R]$. If $K_{\mathrm{O}}(a) = +\infty$ there is nothing to prove. Otherwise, let $(P_a, V_a)$ attain the minimum defining $K_{\mathrm{O}}(a)$, so that $I(P_a,V_a) \geq a$ and $D(V_a \| W | P_a) = K_{\mathrm{O}}(a)$. Since $I(P_a,V_a) \geq a$ implies $|R - I(P_a,V_a)|^{+} \leq R - a$, we write
\begin{align*}
    K_{\mathrm{DK}}(R) & \leq D(V_a \| W | P_a) + |R - I(P_a,V_a)|^{+} \\
    & \leq K_{\mathrm{O}}(a) + (R - a),
\end{align*}
which holds for all $a \in [0,R]$. This completes the proof.
\end{proof}

Lemma~\ref{lem:inf_conv} is a strong converse analogue of \cite[Lemma~10.4]{Csiszar2011}, which treats the error exponent regime (see also \cite[Eq.~(10.27)]{Csiszar2011}).
Its utility for our purposes is the following.
\begin{proposition}
\label{thm:KDK_KO}
    For every $R\ge0$, 
    \begin{equation}\label{eq:KDK_KO_criterion}
\begin{aligned}
& K_{\mathrm{DK}}(R)=K_{\mathrm{O}}(R) \\
\Longleftrightarrow\quad
& K_{\mathrm{O}}(R)
  \le  K_{\mathrm{O}}(a)+ R-a\quad\forall \,a\in [0,R].
\end{aligned}
\end{equation}
\end{proposition}
\begin{proof}
By Lemma~\ref{lem:inf_conv}, $K_{\mathrm{DK}}(R) = K_{\mathrm{O}}(R)$ holds if and only if the term $a = R$ attains the minimum in \eqref{eq:inf_conv}, that is, if and only if $K_{\mathrm{O}}(R) \leq K_{\mathrm{O}}(a) + (R - a)$ for all $a \in [0,R]$.
\end{proof}

In other words, Proposition~\ref{thm:KDK_KO} states that $K_{\mathrm{DK}}(R)$ and $K_{\mathrm{O}}(R)$ are equal if and only if $K_{\mathrm{O}}$ does not increase faster than slope one from any $a \in [0,R]$ to $R$. Since $K_{\mathrm{O}}$ is nondecreasing, this is the same as saying that the two exponents coincide on an initial interval $[0,R]$ if and only if $K_{\mathrm{O}}$ is $1$-Lipschitz on $[0,R]$. 

This is the mirror image of the familiar situation for error exponents, where the random coding exponent $E_{\mathrm{r}}(R)$ and the sphere packing exponent $E_{\mathrm{sp}}(R)$ coincide above the critical rate, beyond which $E_{\mathrm{sp}}(R)$ is $1$-Lipschitz; and differ below it. There, both exponents are convex, so the Lipschitz condition reduces to a condition on the slope of $E_{\mathrm{sp}}$ at $R$. On the contrary, the criterion \eqref{eq:KDK_KO_criterion} requires no convexity of $K_{\mathrm{O}}$.
\subsection{Dual-domain}
\label{sec:dual}
To gain further insight, we now transition to the dual domain. The following equalities hold:
\begin{align}
    & K_{\mathrm{DK}} (R)  = \min_{P,V} \big\{ D(V \| W | P) + |R - I(P,V)|^{+} \big\} \label{eq:dual_def} \\
    & = \min_{P,V} \max_{\rho \in [-1,0]} \big\{ D(V \| W | P) + \rho \big( I(P,V) - R \big) \big\} \label{eq:dual_minmax} \\
    & = \! \max_{\rho \in [-1,0]} \! \big\{ \! - \rho R + \min_{P,V} \left\{ D(V \| W | P) + \rho I(P,V)\right\} \big\}. \label{eq:dual_maxmin}
\end{align}
The first equality is the definition of $K_{\mathrm{DK}}(R)$; the second follows from $|t|^{+} = \max_{\rho \in [-1,0]} \{-\rho t\}$; and the third is Sion's minimax theorem, whose conditions are met since the objective is affine in $\rho$ and, for every $\rho \in [-1,0]$, it is also a convex function of $P \times V$.
Focusing on the inner minimization, it is known that, for every $\rho \in [-1,0]$,
\begin{align}
\nonumber
    \min_{P,V} \left\{ D(V \| W | P) + \rho I(P,V)\right\} & = \min_P E_0(\rho,P) \\
    & = \rho \, C_{\frac{1}{1+\rho}}(W),
    \label{eq:inner_min}
\end{align}
where $C_{\alpha}(W)$, for $\alpha \in [1,\infty]$, is the $\alpha$-capacity (R\'enyi capacity of order $\alpha$) of the DMC $W$; see \cite{Csiszar1995,Nakiboglu2019}.

In view of the above, $K_{\mathrm{DK}}(R) = K_{\mathrm{A}}(R)$ can be written as
\begin{equation}
    K_{\mathrm{DK}}(R) = \max_{\rho \in [-1,0]} \rho \big( C_{\frac{1}{1+\rho}}(W) - R\big),
    \label{eq:dual_representation}
\end{equation}
which was also established in \cite{MosonyiOgawa2017} in the more general setting of classical-quantum channels.
Two consequences of \eqref{eq:dual_representation} are relevant to what follows. 
First, $K_{\mathrm{DK}}$ is convex on $[0,\infty)$, being a maximum of affine functions of $R$ (convexity is proved by an operational argument in \cite{oohama2012}). Second, by \eqref{eq:dual_representation}, $K_{\mathrm{DK}}$ is the upper envelope of the lines
\begin{equation}
    K_{\frac{1}{1+\rho}}(R) := \rho \big( C_{\frac{1}{1+\rho}}(W) - R \big), \qquad \rho \in [-1,0],
    \label{eq:K_rho}
\end{equation}
of slopes $-\rho \in [0,1]$. Hence $K_{\mathrm{DK}}$ has slopes in $[0,1]$, and a maximizing $\rho$ in \eqref{eq:dual_representation} gives a subgradient $-\rho$ at $R$.
Going forward, the steepest line $K_{\infty}(R) = R - C_{\infty}(W)$ plays a crucial role, and we call it the {\em endpoint line}.
\subsection{Endpoint line}
\label{sec:endpoint}
At this point, we know that $K_{\infty}(R) \leq K_{\mathrm{DK}}(R) \leq K_{\mathrm{O}}(R)$ for every $R \geq 0$.
Moreover, since $K_{\mathrm{DK}}$ is convex with slopes in $[0,1]$, once it reaches slope one it follows the endpoint line $K_{\infty}$; whereas by Proposition~\ref{thm:KDK_KO}, a gap opens where $K_{\mathrm{O}}$ increases faster than slope one. This suggests studying where each exponent meets the endpoint line.

Note that at $\rho = -1$, the quantities $E_0(-1,P)$ and $C_{\infty}(W)$ are defined by the corresponding limits. Arimoto \cite{Arimoto1973} identified $\min_P E_0(-1,P) = - \log \sum_{y} \max_{x} W(y|x)$, which implies
\begin{equation}
    C_{\infty}(W) = \log \sum_{y} \max_{x} W(y|x).
    \label{eq:C_inf}
\end{equation}
For $P$ with full support, $C_\infty(W)$ is also known as the {\em maximal leakage} from $X$ to $Y$~\cite{IssaWagnerKamath2020}.
To determine where the two exponents meet the endpoint line $K_{\infty}(R)$, we re-derive \eqref{eq:C_inf} from \eqref{eq:inner_min} by a direct argument. The main purpose is to identify the pairs $(P,V)$ that attain the minimum of $D(V \| W | P) - I(P,V)$.
Let
\begin{equation*}
    Q^{\star}(y) := \mathrm{e}^{-C_{\infty}(W)} \max_{x} W(y|x),
\end{equation*}
a probability distribution on $\mathcal{Y}$. Moreover, for each $x \in \mathcal{X}$, let
\begin{equation}
    \mathcal{Y}_x := \big\{ y \in \mathcal{Y} : W(y|x) = \max\nolimits_{x'} W(y|x') \big\}
    \label{eq:ML_outputs}
\end{equation}
be the set of outputs for which $x$ is a maximum-likelihood input. Consider the set of pairs
\begin{equation}
\nonumber
    \mathcal{M} := \big\{ (P,V) : PV = Q^{\star} \ \text{and} \ 
     P(x)V(y|x) > 0 \Rightarrow y \in \mathcal{Y}_x \big\}.
\end{equation}
The set $\mathcal{M}$ is non-empty, since for any map $\phi : \mathcal{Y} \to \mathcal{X}$ with $W(y|\phi(y)) = \max_{x} W(y|x)$ for all $y$, the pair $(P,V)$ with $P(x)V(y|x) = Q^{\star}(y) \, \mathbbm{1}\{\phi(y) = x\}$ belongs to $\mathcal{M}$.
\begin{lemma}
\label{lem:endpoint_line}
For every pair $(P,V)$,
\begin{equation}
    D(V \| W | P) - I(P,V) \geq - C_{\infty}(W),
    \label{eq:endpoint_bound}
\end{equation}
with equality if and only if $(P,V) \in \mathcal{M}$.
\end{lemma}
\begin{proof}
Expanding $I(P,V)$ and $D(V \| W | P)$, we get
\begin{align}
    &D(V \| W | P) - I(P,V)  \notag\\
      &\; =- H(PV) - \sum_{x,y} P(x) V(y|x) \log W(y|x).
    \label{eq:g_expression}
\end{align}
By the Gibbs inequality, $-H(PV) \geq \sum_{y} PV(y) \log Q^{\star}(y)$, with equality if and only if $PV = Q^{\star}$. Substituting into \eqref{eq:g_expression},
\begin{equation}
    D(V \| W | P) - I(P,V) \geq \sum_{x,y} P(x) V(y|x) \log \frac{Q^{\star}(y)}{W(y|x)}.
    \label{eq:step_gibbs}
\end{equation}
For every $(x,y)$ with $P(x)V(y|x) > 0$,
\begin{equation}
    \frac{Q^{\star}(y)}{W(y|x)} = \mathrm{e}^{-C_{\infty}(W)} \frac{\max_{x'} W(y|x')}{W(y|x)} \geq \mathrm{e}^{-C_{\infty}(W)},
    \label{eq:ratio_bound}
\end{equation}
with equality if and only if $y \in \mathcal{Y}_x$. Substituting \eqref{eq:ratio_bound} into \eqref{eq:step_gibbs} gives \eqref{eq:endpoint_bound}. Equality in \eqref{eq:endpoint_bound} holds if and only if $PV = Q^{\star}$ and equality holds in \eqref{eq:ratio_bound} for every $(x,y)$ with $P(x)V(y|x) > 0$, which is the definition of $\mathcal{M}$.
\end{proof}

Lemma~\ref{lem:endpoint_line} establishes \eqref{eq:inner_min} at $\rho = -1$ and identifies $\mathcal{M}$ as the set of its minimizers. Moreover, by \eqref{eq:g_expression}, $D(V \| W | P) - I(P,V)$ is a convex function of $P \times V$ (this is the case $\rho = - 1$ of the convexity in \eqref{eq:dual_maxmin}). Finally, every $(P,V) \in \mathcal{M}$ satisfies
\begin{equation}
    I(P,V) = D(V \| W | P) + C_{\infty}(W).
    \label{eq:I_on_M}
\end{equation}
\begin{lemma}
\label{lem:interval}
Viewed as a set of joint distributions $P\times V$, $\mathcal M$ is convex and compact.  Moreover, $\{ I(P,V) : (P,V) \in \mathcal{M} \}$ is equal to the interval $ [R_{\star}, R^{\star}]$, where 
\begin{equation}
    R_{\star} := \min_{(P,V) \in \mathcal{M}} I(P,V), \qquad
    R^{\star} := \max_{(P,V) \in \mathcal{M}} I(P,V).
    \label{eq:R_star_def}
\end{equation}
Finally, $C_{\infty}(W) \leq R_{\star} \leq R^{\star} \leq \min\{ \log|\mathcal{X}|, H(Q^{\star}) \}$.
\end{lemma}
\begin{proof}
The two conditions in the definition of $\mathcal{M}$ are linear constraints on $P \times V$, so $\mathcal{M}$ is convex and compact. Since $\mathcal{M}$ is also non-empty and $I(P,V)$ is a continuous function of $P \times V$, the image of $\mathcal{M}$ under $I(\fndot,\fndot)$ is a closed interval, whose endpoints are $R_{\star}$ and $R^{\star}$ by \eqref{eq:R_star_def}. 

Finally, \eqref{eq:I_on_M} gives $I(P,V) \geq C_{\infty}(W)$ on $\mathcal{M}$, while $I(P,V) \leq \log|\mathcal{X}|$ and $I(P,V) \leq  H(Q^{\star})$ on $\mathcal{M}$.
\end{proof}
\section{Main Result}
\label{sec:main}
We now determine where each exponent meets the endpoint line, and settle the question of their equality.
\begin{theorem}
\label{thm:main}
For every $R \geq 0$,
\begin{list}{}{\setlength{\leftmargin}{1.6em}\setlength{\labelwidth}{1.2em}\setlength{\labelsep}{0.4em}\setlength{\itemsep}{0pt}\setlength{\parsep}{0pt}\setlength{\topsep}{2pt}}
    \item[(a)] ${K_{\mathrm{DK}}(R) = K_{\infty}(R)}$ if and only if ${R \geq R_{\star}}$;
    \item[(b)] ${K_{\mathrm{O}}(R) = K_{\infty}(R)}$ if and only if ${R_{\star} \leq R \leq R^{\star}}$;
    \item[(c)] ${K_{\mathrm{O}}(R) = K_{\mathrm{DK}}(R)}$ if and only if ${R \leq R^{\star}}$.
\end{list}
Consequently, ${K_{\mathrm{O}}(R) > K_{\mathrm{DK}}(R) = K_{\infty}(R)}$ for ${R > R^{\star}}$.
\end{theorem}

Theorem~\ref{thm:main} is illustrated schematically in
Fig.~\ref{fig:strong_converse_rates} for the case $R_\star<R^\star<\log|\mathcal X|$. Both exponents follow the endpoint line on $[R_{\star}, R^{\star}]$. Beyond $R^{\star}$, $K_{\mathrm{DK}}$ continues along the line, while $K_{\mathrm{O}}$ lies strictly above it. 

In terms of the inf-convolution identity \eqref{eq:inf_conv}, the theorem identifies the minimizing $a$. For $R \leq R^{\star}$, the minimum is attained at $a = R$. For $R \geq R_{\star}$, it is attained at every $a$ in $ [R_{\star}, \min\{R, R^{\star}\}]$, since there $K_{\mathrm{O}}(a) + (R - a) = K_{\infty}(R)$ by (a) and (b). In particular, for $R > R^{\star}$ the criterion \eqref{eq:KDK_KO_criterion} fails with $a = R^{\star}$. 
We now proceed with the proof of Theorem \ref{thm:main}.
\begin{figure}[!t]
    \centering

\begingroup
\definecolor{SCrateInk}{HTML}{192938}
\definecolor{SCrateMuted}{HTML}{65717D}
\definecolor{SCrateBlue}{HTML}{205E95}
\definecolor{SCrateOrange}{HTML}{C26020}
\definecolor{SCrateGreen}{HTML}{39805A}
\definecolor{SCrateGap}{HTML}{F9E8D9}

\def\SCcapacity{0.90}
\def\SClower{2.50}
\def\SCupper{3.60}
\def\SCrateLimit{4.80}
\def\SCtop{4.40}
\def\SCgapBend{0.80}
\def\SCgapArrow{4.62}
\pgfmathsetmacro{\SCcinfty}{(\SCcapacity+\SClower)/2}
\pgfmathsetmacro{\SClowerY}{\SClower-\SCcinfty}
\pgfmathsetmacro{\SCupperY}{\SCupper-\SCcinfty}
\pgfmathsetmacro{\SClimitDK}{\SCrateLimit-\SCcinfty}
\pgfmathsetmacro{\SClimitKO}{\SCrateLimit-\SCcinfty
    +\SCgapBend*(\SCrateLimit-\SCupper)^2}

\begin{tikzpicture}[
    x=1.42cm,y=0.80cm,
    font=\footnotesize,
    text=SCrateInk,
    >={Latex[length=1.6mm]},
    every node/.style={inner sep=1.5pt},
    dk/.style={SCrateBlue,line width=1.3pt},
    omura/.style={SCrateOrange,line width=0.9pt,
                  dash pattern=on 3.5pt off 2.5pt},
    endpoint/.style={SCrateGreen,line width=0.9pt,
                     dash pattern=on 4pt off 1.6pt on 0.7pt off 1.6pt},
    guide/.style={SCrateMuted!55,line width=0.4pt,densely dashed},
    callout/.style={SCrateMuted,line width=0.5pt,->}
]
    \path[use as bounding box]
    (-0.55,-0.55) rectangle (5.35,{\SCtop+0.2});

    \path[fill=SCrateGap,draw=none]
        (\SCupper,\SCupperY)
        -- plot[domain=\SCupper:\SCrateLimit,samples=80]
            (\x,{\x-\SCcinfty+\SCgapBend*(\x-\SCupper)^2})
        -- (\SCrateLimit,\SClimitDK) -- cycle;

    \draw[->,SCrateInk,line width=0.55pt]
        (0,0) -- ({\SCrateLimit+0.38},0) node[anchor=north west,inner sep=1pt] {$R$};
   \draw[->,SCrateInk,line width=0.55pt]
    (0,0) -- (0,{\SCtop+0.15})
    node[midway,sloped,above=5pt] {Exponent};

    \draw[guide] (\SClower,0) -- (\SClower,\SClowerY);
    \draw[guide] (\SCupper,0) -- (\SCupper,\SCupperY);
    \draw[guide] (\SCrateLimit,0) -- (\SCrateLimit,\SClimitKO);

    \draw[endpoint] (\SCcinfty,0) -- (\SCrateLimit,\SClimitDK);

    \draw[dk] (0,0) -- (\SCcapacity,0);
    \draw[dk] plot[domain=\SCcapacity:\SClower,samples=90]
        (\x,{(\x-\SCcapacity)^2/(2*(\SClower-\SCcapacity))});
    \draw[dk] (\SClower,\SClowerY) -- (\SCrateLimit,\SClimitDK);

    \draw[omura] (0,0) -- (\SCcapacity,0);
    \draw[omura] plot[domain=\SCcapacity:\SClower,samples=90]
        (\x,{(\x-\SCcapacity)^2/(2*(\SClower-\SCcapacity))});
    \draw[omura] (\SClower,\SClowerY) -- (\SCupper,\SCupperY);
    \draw[omura] plot[domain=\SCupper:\SCrateLimit,samples=90]
        (\x,{\x-\SCcinfty+\SCgapBend*(\x-\SCupper)^2});

    \foreach \rate/\label in {
        \SCcapacity/{C},
        \SCcinfty/{C_\infty(W)},
        \SClower/{R_\star},
        \SCupper/{R^\star},
        \SCrateLimit/{\log|\mathcal X|}
    } {
        \draw[SCrateInk,line width=0.45pt] (\rate,-0.06) -- (\rate,0.06);
        \node[anchor=north,fill=white,inner sep=1pt] at (\rate,-0.12) {$\label$};
    }

    \fill[SCrateBlue] (\SCcapacity,0) circle[radius=1.3pt];
    \draw[SCrateGreen,fill=white,line width=0.7pt] (\SCcinfty,0) circle[radius=1.4pt];
    \draw[SCrateBlue,fill=white,line width=0.9pt] (\SClower,\SClowerY) circle[radius=1.8pt];
    \draw[SCrateBlue,fill=white,line width=0.9pt] (\SCupper,\SCupperY) circle[radius=1.8pt];
    \fill[SCrateOrange] (\SCrateLimit,\SClimitKO) circle[radius=1.5pt];
    \fill[SCrateBlue] (\SCrateLimit,\SClimitDK) circle[radius=1.5pt];

    \draw[dk] (0.22,4.10) -- (0.52,4.10);
    \node[anchor=west,text=SCrateBlue] at (0.56,4.10) {$K_{\mathrm{DK}}(R)$};
    \draw[omura] (0.22,3.62) -- (0.52,3.62);
    \node[anchor=west,text=SCrateOrange] at (0.56,3.62) {$K_{\mathrm O}(R)$};
    \draw[endpoint] (0.22,3.14) -- (0.52,3.14);
    \node[anchor=west,text=SCrateGreen] at (0.56,3.14) {$K_{\infty}(R)=R-C_\infty(W)$};

    \node[align=center] (common) at (2.45,2.45)
        {$K_{\mathrm{DK}}=K_{\mathrm O}=K_{\infty}$\\[1pt] on $[R_\star,R^\star]$};
    \draw[callout] (common.south) -- (3.0,{3.0-\SCcinfty+0.06});

    \node[text=SCrateOrange,anchor=south east] at (4.35,3.75) {$K_{\mathrm O}(R)$};
    \path (3.85,{3.85-\SCcinfty}) -- (\SCrateLimit,\SClimitDK)
        node[pos=0.5,sloped,below=3pt,text=SCrateBlue] {$K_{\mathrm{DK}}(R)=K_{\infty}(R)$};
    \draw[<->,SCrateOrange,line width=0.55pt]
        (\SCgapArrow,{\SCgapArrow-\SCcinfty+0.07})
        -- (\SCgapArrow,{\SCgapArrow-\SCcinfty+\SCgapBend*(\SCgapArrow-\SCupper)^2-0.07});
\end{tikzpicture}
\endgroup
    \caption{Schematic illustration of Theorem~\ref{thm:main} when $R_\star<R^\star<\log|\mathcal X|$. The two exponents coincide up to $R^\star$ and share the endpoint line on $[R_\star,R^\star]$. The shaded region indicates the strict gap for $R>R^\star$.}
    \label{fig:strong_converse_rates}
\end{figure}
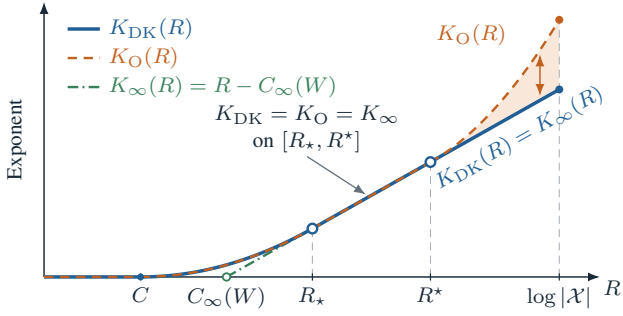
\begin{proof}
We prove each statement separately.

\emph{Part (a).} Fix $(P,V)$. If $I(P,V) \leq R$, then by Lemma~\ref{lem:endpoint_line},
\begin{align*}
    & D(V \| W | P) + |R - I(P,V)|^{+} \\
    & = D(V \| W | P) - I(P,V) + R \geq K_{\infty}(R),
\end{align*}
with equality if and only if $(P,V) \in \mathcal{M}$. On the other hand, if $I(P,V) > R$, then by Lemma~\ref{lem:endpoint_line},
\begin{align*}
    & D(V \| W | P) + |R - I(P,V)|^{+} \\
    & = D(V \| W | P) \geq I(P,V) - C_{\infty}(W) > K_{\infty}(R).
\end{align*}
Minimizing over $(P,V)$ gives $K_{\mathrm{DK}}(R) \geq K_{\infty}(R)$, with equality if and only if $I(P,V) \leq R$ is satisfied by some $(P,V) \in \mathcal{M}$. By Lemma~\ref{lem:interval}, this holds if and only if $R \geq R_{\star}$.

\emph{Part (b).} Fix a pair $(P,V)$ with $I(P,V) \geq R$. By Lemma~\ref{lem:endpoint_line},
\begin{equation*}
    D(V \| W | P) \geq I(P,V) - C_{\infty}(W) \geq K_{\infty}(R),
\end{equation*}
with equality throughout if and only if $(P,V) \in \mathcal{M}$ and $I(P,V) = R$. Minimizing over such pairs gives $K_{\mathrm{O}}(R) \geq K_{\infty}(R)$, with equality if and only if some $(P,V) \in \mathcal{M}$ has $I(P,V) = R$, i.e., by Lemma~\ref{lem:interval}, if and only if $R \in [R_{\star}, R^{\star}]$.

\emph{Part (c).} For $R > R^{\star}$, parts (a) and (b) give $K_{\mathrm{DK}}(R) = K_{\infty}(R) < K_{\mathrm{O}}(R)$, proving the ``only if'' part of the statement. 

Now for the ``if'' part, let $R \leq R^{\star}$. It suffices to show $K_{\mathrm{O}}(R) \leq K_{\mathrm{DK}}(R)$. Let $(P_0,V_0)$ attain the minimum in \eqref{eq:dual_def}. If $I(P_0,V_0) \geq R$, then $(P_0,V_0)$ is feasible in the minimization defining $K_{\mathrm{O}}(R)$, and $K_{\mathrm{O}}(R) \leq D(V_0 \| W | P_0) = K_{\mathrm{DK}}(R)$. Otherwise, we have $I(P_0,V_0) < R$ and
\begin{equation}
    K_{\mathrm{DK}}(R) = D(V_0 \| W | P_0) - I(P_0,V_0) + R.
    \label{eq:K_DK_at_P0}
\end{equation}
By Lemma~\ref{lem:interval}, there is $(P_1,V_1) \in \mathcal{M}$ with $I(P_1,V_1) = R^{\star} \geq R$. For $t \in [0,1]$, let $(P_t,V_t)$ be the pair with
\begin{equation*}
    P_t \times V_t := (1-t) \, (P_0 \times V_0) + t \, (P_1 \times V_1) .
\end{equation*}
By continuity of $I(P_t,V_t)$ in $t$, there is $t \in (0,1]$ with $I(P_t,V_t) = R$. By convexity of $D(V \| W | P) - I(P,V)$ in $P \times V$ and Lemma~\ref{lem:endpoint_line},
\begin{align}
\nonumber
    & D(V_t \| W | P_t) - I(P_t,V_t) \\
\nonumber
    & \leq (1-t) \big( D(V_0 \| W | P_0) - I(P_0,V_0) \big) - t \, C_{\infty}(W) \\
    & \leq D(V_0 \| W | P_0) - I(P_0,V_0).
    \label{eq:segment_convexity}
\end{align}
Adding $I(P_t,V_t) = R$ to both sides and using \eqref{eq:K_DK_at_P0} gives $D(V_t \| W | P_t) \leq K_{\mathrm{DK}}(R)$. Since $(P_t,V_t)$ is feasible in the minimization defining $K_{\mathrm{O}}(R)$, then $K_{\mathrm{O}}(R) \leq K_{\mathrm{DK}}(R)$.
\end{proof}

The proof of Theorem~\ref{thm:main} uses only Lemmas~\ref{lem:endpoint_line} and~\ref{lem:interval}. In particular, it uses neither \eqref{eq:dual_representation} nor the convexity of $K_{\mathrm{DK}}$. The segment argument in part (c) is due to Csisz\'ar \cite[Eq.~(A23)]{Csiszar1995}, who used it for fixed $P$; it extends to variable $P$ by the joint convexity of $D(V \| W | P) - I(P,V)$ in $P \times V$.
\subsection{Critical and Cutoff Rates Above Capacity}
In the dual domain, part (a) states that $\rho = -1$ attains the maximum in \eqref{eq:dual_representation} if and only if $R \geq R_{\star}$. 
By part~(a), $K_{\mathrm{DK}}$ is linear with slope one
on $[R_\star,\infty)$. Its convexity implies that its
right derivative is  less than one for
$R<R_\star$.  Thus $R_{\star}$ plays the role of a critical rate above capacity, while $C_{\infty}(W)$, where $K_\infty(R)$ meets the rate axis, plays the role of a cutoff rate. 


This is analogous to the initial linear branch of the
random coding exponent below capacity, whose rate-axis
intercept is the cutoff rate $R_0=C_{1/2}(W)$.
Note that in the terminology of Csisz\'ar \cite{Csiszar1995}, $C_{1/2}(W)$ and $C_{\infty}(W)$ are the $\beta$-cutoff rates for $\beta = -1$ and $\beta = 1$, respectively.
\subsection{Absence of a Gap}
By part (c) of Theorem~\ref{thm:main} and Lemma~\ref{lem:interval}, $K_{\mathrm{O}}$ and $K_{\mathrm{DK}}$ agree on $[0,\log|\mathcal{X}|]$ if and only if $R^{\star} = \log|\mathcal{X}|$. This condition can be made explicit.
\begin{corollary}
\label{cor:no_gap}
$K_{\mathrm{O}}(R) = K_{\mathrm{DK}}(R)$ for all $R \in [0,\log|\mathcal{X}|]$ if and only if there exists a maximum-likelihood map $\phi$, i.e., $W(y|\phi(y)) = \max_{x} W(y|x)$ for every $y$, such that 
\begin{equation}
    \sum_{y : \phi(y) = x} \max_{x'} W(y|x') = \frac{\mathrm{e}^{C_{\infty}(W)}}{|\mathcal{X}|}
    \label{eq:equal_mass}
\end{equation}
holds for every input $x \in \mathcal{X}$.
\end{corollary}
\begin{proof}
There is no gap if and only if some $(P,V) \in \mathcal{M}$ has $I(P,V) = \log|\mathcal{X}|$. Since $I(P,V) \leq H(P) \leq \log|\mathcal{X}|$, this requires $P$ to be uniform and $I(P,V) = H(P)$, i.e., $X$ is a function of $Y$ under $P \times V$. 

Denoting this function by $\phi$, the support condition in the definition of $\mathcal{M}$ says that $W(y|\phi(y)) = \max_{x} W(y|x)$, and the condition $PV = Q^{\star}$ with $P$ uniform says that 
\begin{equation*}
  \sum_{y : \phi(y) = x} Q^{\star}(y) = \frac{1}{|\mathcal{X}|}, \ \text{for every } x,
\end{equation*}
which is \eqref{eq:equal_mass}. 
Conversely, given such a map $\phi$, the pair $(P,V)$ with $P(x)V(y|x) = Q^{\star}(y) \, \mathbbm{1}\{\phi(y) = x\}$ lies in $\mathcal{M}$, has uniform $P$, and satisfies $I(P,V) = \log|\mathcal{X}|$.
\end{proof}

In words, there is no gap at any rate if and only if some maximum-likelihood map partitions the outputs into $|\mathcal{X}|$ classes with equal weights $\sum_{y : \phi(y) = x} \max_{x'} W(y|x')$. This condition fails for most channels, in which case the two exponents differ on $(R^{\star}, \log|\mathcal{X}|]$; see Section~\ref{sec:examples}.
\begin{remark}
\label{rem:gap_sufficient}
By Lemma~\ref{lem:interval}, a gap is present whenever $H(Q^{\star}) < \log|\mathcal{X}|$, since then we have $R^{\star} < \log|\mathcal{X}|$. This is the case whenever $|\mathcal{Y}| < |\mathcal{X}|$, and for every binary channel with $\max_{x} W(0|x) \neq \max_{x} W(1|x)$.
\end{remark}

\section{Examples}  
\label{sec:examples}
\subsection{Modulo-additive channels}
\label{sec:modulo-additive}
Let $\mathcal{X}=\mathcal{Y}=G$ be a finite abelian group of order $q$, and let $W(y|x)=N(y-x)$ for a noise distribution $N$ on $G$. Define $N_{\max}:=\max_{a\in G}N(a)$ and $A:=\{a\in G:N(a)=N_{\max}\}$, and write $m:=|A|$. Since $\max_x W(y|x)=N_{\max}$ for every $y$, we have $C_\infty(W)=\log(qN_{\max})$, $Q^\star=U_G$, and $\mathcal{Y}_x=x+A$, where $U_G$ denotes the uniform distribution on $G$. Thus $\mathcal{M}$ consists of the pairs $(P,V)$ satisfying $PV=U_G$ and $P(x)V(y|x)>0\Rightarrow y-x\in A$.

Fix any mode $a_0\in A$. The map $\phi(y)=y-a_0$ is an ML map with conditional probability of correct decoding $W(x+a_0|x)=N_{\max}=\mathrm{e}^{C_\infty(W)}/q$ for every  $x$. Thus \eqref{eq:equal_mass} holds, and Corollary~\ref{cor:no_gap} gives
$K_{\mathrm O}(R)=K_{\mathrm{DK}}(R)$ for every $R\in[0,\log q]$.

To determine the lower threshold, note that every $(P,V)\in\mathcal{M}$ satisfies
$I(P,V)=\log q-H(V|P)\ge\log(q/m)$, since each conditional distribution with $P(x)>0$ is supported on the $m$-element set $x+A$. Equality is attained by taking $P=U_G$ and
$V(y|x)=\frac{1}{m}\mathbbm{1}\{y-x\in A\}$, which also gives $PV=U_G$. These observations give
\begin{equation*}
    R_\star=\log(q/m) \quad\mbox{and}\quad R^\star=\log q.
\end{equation*}
On $[R_\star,R^\star]$,
Theorem~\ref{thm:main}(a) and~(b) give
\begin{equation*}
    K_{\mathrm{DK}}(R)=K_{\mathrm O}(R)
        =R-\log(qN_{\max}).
\end{equation*}
Thus the common linear segment has length $\log m$ and is nondegenerate precisely when $N$ has more than one mode.

If the noise distribution has a unique mode, then $R_\star=R^\star=\log q$. Within $[0,\log q]$, the two exponents meet the endpoint line $K_\infty(R)$ only at $R=\log q$, where their common value is $\log(1/N_{\max})$. This includes the binary symmetric channel with crossover
probability different from $1/2$; see Fig.~\ref{fig:bsc}.

\begin{figure}[!t]
    \centering
    \begingroup%
\definecolor{SCrateInk}{HTML}{192938}%
\definecolor{SCrateMuted}{HTML}{65717D}%
\definecolor{SCrateBlue}{HTML}{205E95}%
\definecolor{SCrateOrange}{HTML}{C26020}%
\definecolor{SCrateGreen}{HTML}{39805A}%
\begin{tikzpicture}[
    x=9.6cm,y=12cm,
    font=\footnotesize,
    text=SCrateInk,
    >={Latex[length=1.6mm]},
    every node/.style={inner sep=1.5pt},
    dk/.style={SCrateBlue,line width=1.3pt},
    omura/.style={SCrateOrange,line width=0.9pt,
                  dash pattern=on 3.5pt off 2.5pt},
    endpoint/.style={SCrateGreen,line width=0.9pt,
                  dash pattern=on 4pt off 1.6pt on 0.7pt off 1.6pt},
    guide/.style={SCrateMuted!55,line width=0.4pt,densely dashed},
    callout/.style={SCrateMuted,line width=0.5pt,->}
]
    \def\BSCp{0.25}
    \edef\BSCcapacity{\fpeval{ln(2)+\BSCp*ln(\BSCp)
        +(1-\BSCp)*ln(1-\BSCp)}}
    \edef\BSCcinfty{\fpeval{ln(2*(1-\BSCp))}}
    \edef\BSClimit{\fpeval{ln(2)}}
    \edef\BSCend{\fpeval{-ln(1-\BSCp)}}
    \path[use as bounding box] (-0.135,-0.073) rectangle (0.765,0.357);

    \draw[->,SCrateInk,line width=0.55pt]
        (0,0) -- (0.747,0) node[anchor=south east,yshift=2pt] {$R$};
    \draw[->,SCrateInk,line width=0.55pt]
        (0,0) -- (0,0.335);
    \node[rotate=90] at (-0.10,0.1675) {Exponent};
    \node[anchor=north east] at (-0.004,-0.002) {$0$};

    \draw[guide] (0,\BSCend) -- (\BSClimit,\BSCend);
    \draw[guide] (\BSClimit,0) -- (\BSClimit,\BSCend);
    \draw[SCrateInk,line width=0.45pt]
        (-0.006,\BSCend) -- (0.006,\BSCend);
    \node[anchor=east] at (-0.012,\BSCend) {$\log(4/3)$};
    \foreach \v in {0.1,0.2} {
        \draw[SCrateInk,line width=0.4pt] (-0.005,\v) -- (0.005,\v);
        \node[anchor=east,text=SCrateMuted,font=\scriptsize]
            at (-0.010,\v) {$\v$};
    }

    \draw[endpoint] (\BSCcinfty,0) -- (\BSClimit,\BSCend);

    \draw[dk] (0,0) -- (\BSCcapacity,0);
    \draw[dk,line join=round]
        plot[domain=0.25:0.00001,samples=220,variable=\t]
        ({\fpeval{ln(2)+\t*ln(\t)+(1-\t)*ln(1-\t)}},
         {\fpeval{\t*ln(\t/\BSCp)
             +(1-\t)*ln((1-\t)/(1-\BSCp))}})
        -- (\BSClimit,\BSCend);

    \draw[omura] (0,0) -- (\BSCcapacity,0);
    \draw[omura,line join=round]
        plot[domain=0.25:0.00001,samples=220,variable=\t]
        ({\fpeval{ln(2)+\t*ln(\t)+(1-\t)*ln(1-\t)}},
         {\fpeval{\t*ln(\t/\BSCp)
             +(1-\t)*ln((1-\t)/(1-\BSCp))}})
        -- (\BSClimit,\BSCend);

    \foreach \rate in {\BSCcapacity,\BSCcinfty,\BSClimit} {
        \draw[SCrateInk,line width=0.45pt]
            (\rate,-0.006) -- (\rate,0.006);
    }
    \node[anchor=north] at (\BSCcapacity,-0.012) {$C$};
    \node[anchor=north] at (\BSCcinfty,-0.012) {$C_\infty(W)$};
    \node[anchor=north east,align=center] at (0.738,-0.012)
        {$R_\star=R^\star$\\[-1pt]$=\log|\mathcal X|=\log 2$};

    \fill[SCrateBlue] (\BSCcapacity,0) circle[radius=1.4pt];
    \draw[SCrateGreen,fill=white,line width=0.7pt]
        (\BSCcinfty,0) circle[radius=1.4pt];
    \draw[SCrateBlue,fill=white,line width=0.9pt]
        (\BSClimit,\BSCend) circle[radius=1.9pt];

    \draw[dk] (0.047,0.311) -- (0.095,0.311);
    \node[anchor=west,text=SCrateBlue] at (0.105,0.311)
        {$K_{\mathrm{DK}}(R)$};
    \draw[omura] (0.047,0.267) -- (0.095,0.267);
    \node[anchor=west,text=SCrateOrange] at (0.105,0.267)
        {$K_{\mathrm O}(R)$};
    \draw[endpoint] (0.047,0.230) -- (0.095,0.230);
    \node[anchor=west,text=SCrateGreen] at (0.105,0.230)
        {$K_\infty(R)=R-C_\infty(W)$};

    \node[align=center] (common) at (0.285,0.141)
        {$K_{\mathrm{DK}}=K_{\mathrm O}$\\[1pt]on $[0,\log 2]$};
    \draw[callout,shorten >=2pt] (common.east) -- (0.4662,0.1252);
    \node[align=center,font=\scriptsize] (contact) at (0.578,0.333)
        {Single contact\\with slope $1$};
    \draw[callout] (contact.south east) -- (0.687,0.293);
\end{tikzpicture}%
\endgroup
    \caption{Strong converse exponents of the BSC
    with crossover probability $1/4$. }
    \label{fig:bsc}
\end{figure}
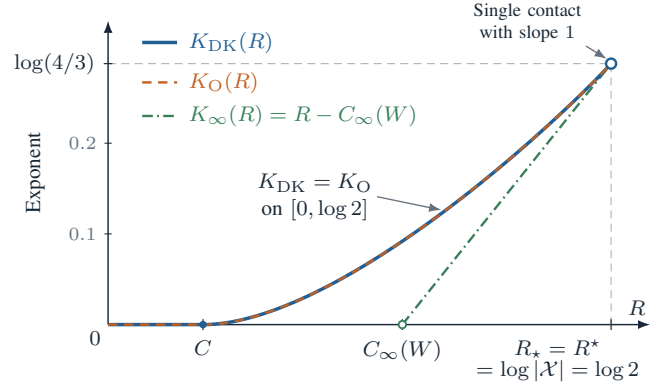

\subsection{The binary erasure channel (BEC)}
\label{sec:bec}
Let $\mathcal{X}=\{0,1\}$ and $\mathcal{Y}=\{0,1,e\}$, with $W(x|x)=1-\epsilon$ and $W(e|x)=\epsilon$ for $\epsilon\in(0,1)$. For $y\in\{0,1\}$, the column maximum is $1-\epsilon$, attained uniquely at $x=y$. For the erasure output, the maximum is $\epsilon$, attained at both inputs. Hence, $C_\infty(W)=\log(2-\epsilon)$ and $Q^\star=(1-\epsilon,1-\epsilon,\epsilon)/(2-\epsilon)$, with $\mathcal{Y}_0=\{0,e\}$ and $\mathcal{Y}_1=\{1,e\}$.

Each pair $(P,V)\in\mathcal{M}$ is determined by the fraction $\beta\in[0,1]$ of the erasure mass $Q^\star(e)$ assigned to input $0$. Its joint distribution is
\begin{equation*}
    P\times V=\frac{1}{2-\epsilon}
    \begin{bmatrix}
        1-\epsilon & 0 & \beta\epsilon\\
        0 & 1-\epsilon & (1-\beta)\epsilon
    \end{bmatrix}.
\end{equation*}
Since $PV=Q^\star$ for every $\beta$, $I(P,V)=H(Q^\star)-H(V|P)$. The conditional entropy is strictly concave and symmetric in $\beta$, so it is maximized at $\beta=1/2$ and minimized at $\beta\in\{0,1\}$.

At $\beta=1/2$, $P$ is uniform and $V$ is a BEC with erasure probability $\frac{\epsilon}{2-\epsilon}$. At $\beta=1$, the conditional supports of $V$ are disjoint and $P(0)=\frac{1}{ 2-\epsilon}$, so $I(P,V)=H(P)$. Consequently,
\begin{equation*}
    R_\star=\frac{2(1-\epsilon)}{2-\epsilon}\log 2,
    \qquad
    R^\star=h\!\left(\frac{1}{2-\epsilon}\right),
\end{equation*}
where $h$ denotes the binary entropy function. For every $\epsilon\in(0,1)$, $R_\star<R^\star<\log 2$.

The equal-mass condition in Corollary~\ref{cor:no_gap} also fails directly: every ML map assigns outputs $0$ and $1$ to their corresponding inputs and the erasure symbol to one of them. The resulting decoding classes have weights $1$ and $1-\epsilon$, which are unequal.

By Theorem~\ref{thm:main}(c), the two exponents coincide for every $R\in[0,R^\star]$. On $[R_\star,R^\star]$, parts~(a) and~(b) give
\begin{equation*}
    K_{\mathrm{DK}}(R)=K_{\mathrm O}(R)
        =R-\log(2-\epsilon).
\end{equation*}
For $R\in(R^\star,\log 2]$, $K_{\mathrm{DK}}(R)$ continues along this line, while $K_{\mathrm O}(R)>K_{\mathrm{DK}}(R)$.

At $R=\log 2$, the constraint $I(P,V)\ge\log 2$ forces $P$ to be uniform and the conditional supports of $V$ to be disjoint. Finite divergence restricts these supports to those of $W$, so at least one row must put all its mass on its non-erasure output. Minimizing the divergence leaves the other row equal to the corresponding row of $W$. Thus,
$K_{\mathrm O}(\log 2)=\frac12\log\frac{1}{1-\epsilon}$
and
$K_{\mathrm{DK}}(\log 2)=\log\frac{2}{2-\epsilon}$.
The gap is
$\frac12\log\bigl(1+\frac{\epsilon^2}{4(1-\epsilon)}\bigr)>0$.
\section{Concluding Remarks}

Our characterization identifies two distinct transition rates: $R_\star$, at which the Dueck--K\"orner exponent $K_{\mathrm{DK}}$ reaches the
endpoint line $K_{\infty}$, and $R^\star$, beyond which Omura's exponent $K_{\mathrm{O}}$  is strictly larger. As shown for the BEC, these rates can delimit a nondegenerate interval on which both exponents coincide with the endpoint line. This contrasts with the random coding and sphere packing exponents below capacity, which cannot share a nondegenerate linear segment of slope $-1$.

Moreover, for DMCs satisfying the equal-mass condition in Corollary~\ref{cor:no_gap}, the two strong converse exponents coincide throughout $[0,\log|\mathcal X|]$. Such equality throughout the below-capacity rate range is impossible for the random coding and sphere packing exponents of a positive-capacity DMC, since they differ at sufficiently small positive rates.

\section*{Acknowledgments}
The authors thank Prof.\ Bar{\i}{\c{s}} Nakibo{\u{g}}lu for identifying an error in a preprint of their earlier paper~\cite{Tan2025}, which motivated this work. They also thank Prof.\ Yasutada Oohama for clarifications concerning his paper~\cite{oohama2012}.

The authors used OpenAI Codex (GPT-6 Astra, Ultra mode) and Anthropic's Claude (Opus 5 and Fable 5.1) to explore and stress-test proof strategies, identify counterexamples, formulate and refine intermediate lemmas, check algebraic and probabilistic calculations, and improve the exposition. The authors independently verified all mathematical statements, proofs, and references in the final manuscript and take full responsibility for its contents.
\bibliographystyle{IEEEtran}
\bibliography{References}

@article{Wolfowitz1957,
  title={The coding of messages subject to chance errors},
  author={Wolfowitz, Jacob},
  journal={Illinois J. Math.},
  volume={1},
  number={4},
  pages={591--606},
  year={1957}
}

@article{IssaWagnerKamath2020,
  author  = {Ibrahim Issa and Aaron B. Wagner and Sudeep Kamath},
  title   = {An Operational Approach to Information Leakage},
  journal={IEEE Trans. Inf. Theory}, 
  volume  = {66},
  number  = {3},
  pages   = {1625--1657},
  month   = mar,
  year    = {2020},
  doi     = {10.1109/TIT.2019.2962804}
}

@inproceedings{ChengNakiboglu2020,
  author    = {Hao-Chung Cheng and Bar{\i}{\c{s}} Nakibo{\u{g}}lu},
  title     = {Refined Strong Converse for the Constant Composition Codes},
  booktitle = {Proc. IEEE Int. Symp. Inf. Theory (ISIT)},
  pages     = {2149--2154},
  year      = {2020},
  doi       = {10.1109/ISIT44484.2020.9174315}
}

@ARTICLE{Arimoto1973,
  author={Arimoto, S.},
  journal={IEEE Trans. Inf. Theory}, 
  title={On the converse to the coding theorem for discrete memoryless channels (Corresp.)}, 
  year={1973},
  volume={19},
  number={3},
  pages={357-359},
  doi={10.1109/TIT.1973.1055007}}

@article{Dueck1979,
  title={Reliability function of a discrete memoryless channel at rates above capacity (Corresp.)},
  author={Dueck, G. and K\"orner, J.},
  journal={IEEE Trans. Inf. Theory},
  volume={25},
  number={1},
  pages={82--85},
  year={1979},
  publisher={IEEE}
}

@article{Omura1975,
  title={A lower bounding method for channel and source coding probabilities},
  author={Omura, Jim K},
  journal={Inf. Control},
  volume={27},
  number={2},
  pages={148--177},
  year={1975},
  publisher={Elsevier}
}

@book{Gallager1968,
  title={Information theory and reliable communication},
  author={Gallager, Robert G},
  year={1968},
  publisher = {John Wiley \& Sons, Inc.},
  address   = {New York}
}

@article{Haroutunian1968,
  author  = {Haroutunian, E. A.},
  title   = {Estimates of the error exponent for the semi-continuous memoryless channel},
  journal = {Probl. Peredachi Inf.},
  volume  = {4},
  number  = {4},
  pages   = {37--48},
  year    = {1968},
  note    = {In Russian}
}

@book{Csiszar2011,
  title={Information theory: {Coding} theorems for discrete memoryless systems},
  author={Csisz{\'a}r, Imre and K{\"o}rner, J{\'a}nos},
  year={2011},
  publisher={Cambridge University Press}
}

@INPROCEEDINGS{oohama2012,
  author={Oohama, Y.},
  booktitle={Proc. IEEE Int. Symp. Inf. Theory and Appl. (ISITA)}, 
  title={On two strong converse theorems for discrete memoryless channels}, 
  year={2012},
  volume={},
  number={},
  pages={46--50}}

@ARTICLE{Tan2025,
  author={Tan, Vincent Y. F. and Joudeh, Hamdi},
  journal={IEEE Trans. Inf. Theory}, 
  title={Ensemble-Tight Second-Order Asymptotics and Exponents for Guessing-Based Decoding With Abandonment}, 
  year={2025},
  volume={71},
  number={10},
  pages={7555-7567},
  doi={10.1109/TIT.2025.3597405}}

@article{Csiszar1995,
  author  = {Csisz{\'a}r, Imre},
  title   = {Generalized cutoff rates and {R}{\'e}nyi's information measures},
  journal = {IEEE Trans. Inf. Theory},
  volume  = {41},
  number  = {1},
  pages   = {26--34},
  year    = {1995}
}

@article{Nakiboglu2019,
  author  = {Nakibo{\u{g}}lu, Bar{\i}{\c{s}}},
  title   = {The {A}ugustin capacity and center},
  journal = {Probl. Inf. Transm.},
  volume  = {55},
  number  = {4},
  pages   = {299--342},
  year    = {2019}
}

@article{MosonyiOgawa2017,
  author  = {Mosonyi, Mil{\'a}n and Ogawa, Tomohiro},
  title   = {Strong converse exponent for classical-quantum channel coding},
  journal = {Commun. Math. Phys.},
  volume  = {355},
  number  = {1},
  pages   = {373--426},
  year    = {2017}
}

@INPROCEEDINGS{Polyanskiy2010,
  author={Polyanskiy, Yury and Verdú, Sergio},
  booktitle={Proc. Allerton Conf. Commun., Control, Comput.}, 
  title={Arimoto channel coding converse and Rényi divergence}, 
  year={2010},
  volume={},
  number={},
  pages={1327-1333},
  doi={10.1109/ALLERTON.2010.5707067}}

\end{document}